\documentclass[11pt]{article} 
\usepackage{fullpage} 
\usepackage{amsfonts} 
\usepackage{algorithmic}
\usepackage{algorithm}
\usepackage{amsmath,amsthm,amssymb} 
\usepackage{latexsym}
\usepackage{epic}
\usepackage{epsfig}
\usepackage{amscd}
\usepackage{url}
\usepackage{verbatim}

\newtheorem{theorem}{Theorem}[section]

\newtheorem{lemma}[theorem]{Lemma}

\newtheorem{claim}[theorem]{Claim}
\newtheorem{definition}[theorem]{Definition}
\newtheorem{remark}[theorem]{Remark}

\newtheorem{fact}[theorem]{Fact}

\newtheorem{thm}[theorem]{Theorem}

\newcommand{\by}{\times}

\newcommand{\set}[1]{\left\{ #1 \right\}}
\newcommand{\union}{\cup}
\newcommand{\intersect}{\cap}

\newcommand{\sm}{\setminus}

\renewcommand{\tilde}{\widetilde}

\DeclareMathOperator{\poly}{poly}

\def\pr{\qopname\relax n{Pr}}
\def\ex{\qopname\relax n{E}}

\def\min{\qopname\relax n{min}}
\def\max{\qopname\relax n{max}}

\def\Ex{\qopname\relax n{\mathbf{E}}}

\newcommand{\expect}[2][]{\ex_{#1} [#2]}

\newcommand{\RR}{\mathbb{R}}
\newcommand{\RRp}{\RR^+}

\newcommand{\NN}{\mathbb{N}}

\def\A{\mathcal{A}}

\def\D{\mathcal{D}}

\def\H{\mathcal{H}}
\def\I{\mathcal{I}}

\def\L{\mathcal{L}}
\def\P{\mathcal{P}}
\def\R{\mathcal{R}}
\def\V{\mathcal{V}}
\def\X{\mathcal{X}}

\def\sse{\subseteq}

\newcommand{\grad}{\bigtriangledown}

\newcommand{\eat}[1]{}

\newcommand{\INPUT}{\item[\textbf{Input:}]}
\newcommand{\OUTPUT}{\item[\textbf{Output:}]}
\newcommand{\PARAMETER}{\item[\textbf{Parameter:}]}

\newcommand{\maxi}[1]{\mbox{maximize} & {#1 } & \\}
\newcommand{\st}{\mbox{subject to} }
\newcommand{\con}[1]{&#1 & \\}
\newcommand{\qcon}[2]{&#1, & \mbox{for } #2.  \\}
\newenvironment{lp}{\begin{equation}  \begin{array}{lll}}{\end{array}\end{equation}}
\newenvironment{lp*}{\begin{equation*}  \begin{array}{lll}}{\end{array}\end{equation*}}

\title{From Convex Optimization to Randomized Mechanisms:\\ Toward Optimal Combinatorial Auctions\footnote{Extended abstract appears in \emph{Proceedings of the 43rd ACM Symposium on Theory of Computing (STOC), 2011.} }}

 \author{
Shaddin Dughmi\thanks{Department of Computer Science, Stanford
 University, 460 Gates Building, 353 Serra Mall, Stanford, CA 94305.  Supported by NSF Grant CCF-0448664 and a Siebel Foundation Scholarship. Email: {\tt shaddin@cs.stanford.edu}.} 
\and 
Tim Roughgarden\thanks{Department of Computer Science, Stanford
 University, 462 Gates Building, 353 Serra Mall, Stanford, CA 94305.
 Supported in part by NSF CAREER Award
 CCF-0448664, an ONR Young Investigator Award, an ONR PECASE Award,
 an AFOSR MURI grant, and an Alfred P. Sloan Fellowship. Email: {\tt
 tim@cs.stanford.edu}.} 
\and
Qiqi Yan\thanks{Department of Computer Science, Stanford
 University, 460 Gates Building, 353 Serra Mall, Stanford, CA 94305. Supported by a Stanford Graduate Fellowship. Email: {\tt
 qiqiyan@cs.stanford.edu}.}
} 

\begin{document}


\maketitle

\thispagestyle{empty} 
\addtocounter{page}{-1} 

\begin{abstract}
We design an expected polynomial-time, truthful-in-expectation,
$(1-1/e)$-approximation mechanism for welfare maximization in a
fundamental class of combinatorial auctions.  Our results apply to
bidders with valuations that are {\em matroid rank sums (MRS)}, which 
encompass most concrete examples of submodular functions studied in 
this context, including coverage functions, matroid weighted-rank
functions, and convex combinations thereof.  Our approximation factor is the best possible, even for
known and explicitly given coverage valuations, assuming $P \neq NP$.
Ours is the first truthful-in-expectation and polynomial-time
mechanism to achieve a constant-factor approximation for an $NP$-hard
welfare maximization problem in combinatorial auctions with 
heterogeneous goods and restricted valuations.

Our mechanism is an instantiation of a new framework for designing
approximation mechanisms based on randomized
rounding algorithms.  A typical such algorithm first optimizes over
a fractional relaxation of the original problem, and then randomly
rounds the fractional solution to an integral one.  With rare
exceptions, such algorithms cannot be converted into
truthful mechanisms.  
The high-level idea of our mechanism design framework is to optimize
{\em directly over the (random) output of the rounding algorithm},
rather than over the {\em input} to the rounding algorithm.
This approach leads to truthful-in-expectation mechanisms, and these
mechanisms can be implemented efficiently when the corresponding
objective function is concave.  For bidders with MRS valuations, we
give a novel randomized rounding algorithm that leads to both
a concave objective function and a $(1-1/e)$-approximation of the
optimal welfare.

\end{abstract}

\newpage 




\section{Introduction}

The overarching goal of {\em algorithmic mechanism design} is to design
computationally efficient algorithms that solve or approximate
fundamental optimization problems in which the underlying data is a
priori unknown to the algorithm.  A central example in both theory and
practice is welfare-maximization in combinatorial auctions.  Here,
there are $m$ items for sale and $n$ bidders vying for them.  Each
bidder $i$ has a private {\em valuation} $v_i(S)$ for each subset~$S$
of the items.\footnote{Each bidder has an exponential number of
  private values; we ignore the attendant representation issues for
the moment.}
The {\em welfare} of an allocation $S_1,\ldots,S_n$ of
the items to the bidders is $\sum_{i=1}^n v_i(S_i)$.  
Since valuations are initially unknown to the seller, computing a
near-optimal allocation requires eliciting information from the
(self-interested) bidders, for example via a bid.  
A {\em mechanism} is a protocol that extracts such information and
computes an allocation of the items and payments.

The ``holy grail'' for a mechanism designer is to devise a
computationally efficient and incentive-compatible mechanism with
an approximation factor that matches the best one known for the (easier)
problem in which the underlying data is provided up front.\footnote{In
  this paper, by ``incentive compatible'' we generally mean 
a (possibly randomized) mechanism such that every
participant maximizes its expected payoff by truthfully
revealing its information to the mechanism, no matter
how the other participants behave. Such mechanisms are
called truthful-in-expectation, and are defined formally
in Section~\ref{sec:MD}.}
Such results are usually difficult to obtain, and in some cases are provably impossible using deterministic mechanisms~\cite{LMN03,PSS08}.
The space of randomized mechanisms, however, is much more promising as shown recently in \cite{DD09,DR10}.\footnote{We note that the impressively general positive
  results for   implementations in Bayes-Nash
  equilibria that were recently obtained in \cite{HL10,soda11a,soda11b} 
  do not apply to the stronger incentive-compatibility notions used in
  this paper and in most of the algorithmic mechanism design literature.} This paper provides such a positive result for a fundamental class of
combinatorial auctions, via a novel randomized mechanism design
framework based on convex optimization.

Algorithmic mechanism design is difficult because 
incentive compatibility severely limits how the algorithm
can compute an outcome, which prohibits use of most of the ingenious
approximation algorithms that have been developed for different
optimization problems.  
More concretely, the only general approach known for designing
(randomized) truthful mechanisms is via {\em maximal-in-distributional
  range (MIDR) algorithms}~\cite{DD09,DR10}. 
An MIDR algorithm fixes a set of distributions over feasible solutions
--- the {\em distributional range} --- independently of the 
valuations reported by the self-interested participants, and outputs a
random sample from the distribution that maximizes expected (reported)
welfare.
The \emph{Vickrey-Clarke-Groves (VCG)} payment scheme renders an MIDR
algorithm truthful-in-expectation.

Most approximation algorithms are not MIDR algorithms.
Consider, as an example, a {\em randomized rounding} algorithm for
welfare maximization in combinatorial auctions (e.g. \cite{F06,DS06}).
We can view such an algorithm as the composition of two algorithms,
a {\em relaxation algorithm} and a {\em rounding algorithm}.
The relaxation algorithm is deterministic and takes as input the
problem data (players' valuations~$v$), and outputs the (fractional)
solution to a linear programming relaxation of the
welfare-maximization problem that is optimal for the objective
function defined by~$v$.
The rounding algorithm is randomized and takes as input this fractional
solution and outputs a feasible allocation of the items to the players.
Taken together, these algorithms assign to each input~$v$ a
probability distribution~$D(v)$ over integral allocations.
For almost all known randomized rounding algorithms, there is an
input~$v$ such that the expected objective function value 
$\expect[y \sim D(v)]{v^Ty}$
with the distribution~$D(v)$ is inferior to that 
$\expect[y \sim D(w)]{v^Ty}$
with a distribution~$D(w)$ that the algorithm would produce for a
different input~$w$ --- and this is a violation of the MIDR property.
Informally, such violations are inevitable unless a rounding algorithm
is designed explicitly to avoid them, on top of the usual
approximation requirements.

The exception that proves the rule is the important and well-known
mechanism design framework of Lavi and Swamy~\cite{LS05}.
Lavi and Swamy~\cite{LS05} begin with the foothold that the 
{\em fractional} welfare maximization problem --- the relaxation algorithm
above --- can be made truthful by charging appropriate VCG payments.  
Further, they identify a very special type of rounding algorithm
that preserves truthfulness: if the expected allocation
produced by the rounding algorithm is {\em always identical to the input to the rounding algorithm},
component-wise, up to some universal scaling factor $\alpha$, then
composing the two algorithms easily yields an $\alpha$-approximate
truthful-in-expectation mechanism (after scaling the fractional VCG
payments by~$\alpha$).
Perhaps surprisingly, there are some interesting problems, such as
welfare maximization in combinatorial auctions with general
valuations, that admit such a rounding algorithm with a best-possible
approximation guarantee (assuming $P \neq NP$).
%
However, most $NP$-hard welfare maximization problems do not seem to
admit good randomized rounding algorithms of the rigid type required
by this design framework.

\subsection{Our Contributions}

We introduce a new approach to designing truthful-in-expectation
approximation mechanisms based on randomized rounding algorithms; we
outline it here for the special case of welfare maximization in
combinatorial auctions.
The high-level idea is to optimize {\em directly on the outcome of the
  rounding algorithm}, rather than merely on the outcome of the
relaxation algorithm (the {\em input} to the rounding algorithm). 
In other words, let $r(x)$ denote a randomized rounding algorithm, from
fractional allocations to integer allocations.
Given players' valuations~$v$, we compute a fractional
allocation~$x$ that maximizes the expected 
welfare $\expect[y \sim
r(x)]{v^Ty}$ over all fractional allocations $x$.
This methodology evidently gives MIDR algorithms.
This optimization problem is often intractable, but
when the rounding algorithm~$r$ and the space of valuations~$v$ are
such that the function $\expect[y \sim r(x)]{v^Ty}$ is always concave in $x$
--- in which case we call~$r$ a {\em convex rounding algorithm} ---
it can be solved in polynomial time using convex
programming (modulo numerical issues that we address later).

We use this design framework to give an expected polynomial-time, truthful-in-expectation,  $(1-1/e)$-approximation mechanism for welfare
maximization in combinatorial auctions in which bidders' valuations
are {\em matroid rank sums (MRS)} --- non-negative linear combinations 
of matroid rank functions on the items.  MRS valuations are submodular
and encompass most concrete examples of submodular functions that have
been studied in the combinatorial auctions literature, including all
coverage functions and matroid weighted-rank functions (see
Section~\ref{sec:MRS} for formal definitions).
Our approximation guarantee is optimal, assuming $P \neq NP$, even
for the special case of the welfare maximization problem with known
and explicitly presented coverage valuations.
Our mechanism is the first truthful-in-expectation and polynomial-time
mechanism to achieve a constant-factor approximation for any $NP$-hard
special case of combinatorial auctions that doesn't assume that there
are multiple copies of every type of item.
It works with ``black-box'' valuations, provided that they support a
randomized analog of a ``value oracle''. We also give a (non-oracle-based) version of the mechanism for explicitly represented coverage valuations.


\subsection{Related Work}

We discuss only the results most pertinent to this work;
see~\cite{CSS05} for an introduction to combinatorial auctions,
and~\cite{BN07} for a survey of truthful approximation mechanisms for
combinatorial auctions.

For the welfare maximization problem in combinatorial auctions with
general valuations (assuming only that $v_i(\emptyset) = 0$ and that
$v_i(S) \le v_i(T)$ whenever $S \subseteq T$), the best
approximation factor possible by a polynomial-time approximation algorithm is
roughly $\min \{ \sqrt{m},n \}$,
where $n$ is the number of bidders and  $m$ is the number of items.
There are comparable 
unconditional lower bounds in various oracle models, assuming polynomial
communication and unbounded computation~\cite{MSV08}; and, assuming that
$P \neq NP$, for various classes of succinctly represented
valuations~\cite{LOS}.

These strong negative results for welfare maximization with general
valuations motivate the study of important special cases.  
Numerous special cases have been considered (see~\cite[Fig 1.2]{BN07}),
and the most well-studied one is for bidders with valuations that are
{\em submodular}, meaning that $v_i(S \cap T) + v_i(S \cup T) \le
v_i(S) + v_i(T)$ for every bidder~$i$ and bundles $S,T$ of items. 
Submodular functions play a fundamental role in combinatorial
optimization, and have a natural economic interpretation in terms of
diminishing marginal returns.

Without incentive-compatibility constraints, the welfare maximization
problem with submodular bidder valuations is completely solved.
Vondr{\'a}k~\cite{V08} gave a $(1-\tfrac{1}{e})$-approximation algorithm
for the problem, improving over the $\tfrac{1}{2}$-approximation given
in~\cite{LLN01}.  The algorithm in~\cite{V08} works in the {\em
  value oracle} model, where each valuation~$v$ is modeled as a
``black box'' that returns the value $v(S)$ of a queried set~$S$ in a
single operation.  The approximation factor of~$1-\tfrac{1}{e}$ is
unconditionally optimal in the value-oracle model (for polynomial
communication)~\cite{MSV08}, and is also optimal (for polynomial time) for certain
succinctly represented submodular valuations, assuming $P \neq NP$~\cite{KhotLMM08}. The result of \cite{KhotLMM08}  implies that $1-1/e$ is the optimal approximation factor in our model as well, assuming $P \neq NP$.\footnote{We show in Appendix~\ref{app:separation} that our oracle model is no more powerful than polynomial-time computation in the special case of explicitly represented coverage functions, for which $1-1/e$ is optimal assuming $P \neq NP$~\cite{KhotLMM08}. In contrast, the work of \cite{FV06} improves on the approximation factor of $1-1/e$ by using \emph{demand oracles}, which can not be simulated in polynomial time for explicit coverage functions.}

Despite intense study, prior to this work, 
there were no truthful-in-expectation and polynomial-time constant-factor
approximation mechanisms for welfare maximization with any non-trivial
subclass of submodular bidder valuations.
The best previous results, which apply to all submodular valuations,
are a truthful-in-expectation  $O\left(\frac{\log m}{\log \log m}\right)$ approximation mechanism in the communication complexity model due to Dobzinski, Fu and Kleinberg~\cite{DFK11}, and a universally-truthful\footnote{A mechanism is universally-truthful if, for \emph{every} realization of a the mechanism's coins, each player maximizes his payoff by bidding truthfully. Universally truthful mechanisms are defined formally in Section \ref{sec:MD}} $O(\log m \log \log m)$ approximation mechanism in the \emph{demand oracle} model due to Dobzinski~\cite{D07}.


The aforementioned works~\cite{DFK11,LS05} are precursors to our general design framework that optimizes directly over the output of a randomized rounding algorithm. In the framework of Lavi and Swamy~\cite{LS05}, the input to and output of the rounding algorithm are assumed to coincide up to a scaling
factor, so optimizing over its input (as they do) is equivalent to optimizing over its output (as we do).  In the result of Dobzinski et al.~\cite{DFK11}, optimizing with respect to their ``proxy bidders'' is equivalent to optimizing over the output of a particular randomized rounding algorithm.

\section{Preliminaries}
\def\S{\mathcal{S}}
\subsection{Optimization Problems}
\label{sec:problem}

We consider  optimization problems $\Pi$ of the following general
form. Each instance of $\Pi$ consists of a \emph{feasible set}  $\S$, and an
\emph{objective function} $w:\S \to \RR$.  
The solution to an instance of $\Pi$
is given by the following optimization problem. 
\begin{lp}\label{lp:problem}
  \maxi{w(x)}
\st
\con{x \in \S.}
\end{lp}

\subsection{Mechanism Design Basics}\label{sec:MD}
We consider mechanism design optimization problems of the form in \eqref{lp:problem}. In such problems, there are $n$ players, where each player $i$ has a \emph{valuation function} $v_i: \S \to \RR$. We are concerned with \emph{welfare maximization} problems, where the objective is $w(x) = \sum_{i=1}^n v_i(x)$.

We consider direct-revelation mechanisms for 
optimization mechanism design problems.  Such a mechanism comprises an {\em
  allocation rule}, which is a function from (hopefully
truthfully) reported valuation functions $v_1,\ldots,v_n$ to an outcome
$x \in \S$, and a {\em payment rule}, which is a function from
reported valuation functions to a required payment from each player.
We allow the allocation and payment rules to be randomized.

A mechanism with allocation and payment rules $\A$ and $p$ is {\em
  truthful-in-expectation} if every player always maximizes its expected
  payoff by truthfully reporting its valuation function, meaning that
\begin{equation}\label{eq:truthful}
\Ex[v_i(\A(v)) - p_i(v)] \geq \Ex[ v_i(\A(v'_i,v_{-i})) -
  p_i(v'_i,v_{-i})]
\end{equation}
for every player~$i$, (true) valuation function~$v_i$, (reported)
valuation function~$v'_i$, and (reported) valuation functions~$v_{-i}$ of
the other players.
The expectation in~\eqref{eq:truthful} is over the coin flips of the
mechanism.  
If \eqref{eq:truthful} holds for every flip of the coins, rather than merely in expectation, we call the mechanism \emph{universally truthful}.

The mechanisms that we design can be thought of as randomized
variations on the classical VCG mechanism, as we explain next.  
Recall that the {\em VCG
  mechanism} is defined by the (generally intractable)
allocation rule that selects the welfare-maximizing outcome with
respect to the reported valuation functions, and the payment rule that
charges each player~$i$ a bid-independent ``pivot term'' minus the
reported welfare earned by other players in the selected outcome.  This
(deterministic) mechanism is truthful; see e.g.~\cite{Nis07}.

Now let $dist(\S)$ denote the probability distributions over a
feasible set $\S$, and let $\D \sse dist(\S)$ be a compact subset of
them.  The corresponding {\em Maximal in Distributional Range (MIDR)}
allocation rule is defined as follows: given reported valuation
functions $v_1,\ldots,v_n$, return an outcome that is sampled randomly
from a distribution $D^* \in \D$ that maximizes the expected welfare
$\Ex_{x \sim D}[\sum_i v_i(x)]$ over all distributions $D \in \D$.
Analogous to the VCG mechanism, there is a (randomized) payment rule
that can be coupled with this allocation rule to yield a
truthful-in-expectation mechanism (see~\cite{DD09}).

\subsection{Combinatorial Auctions}

In \emph{Combinatorial Auctions} there is a set
$[m]=\set{1,2,\ldots,m}$ of items, and a set
$[n]=\set{1,2,\ldots,n}$ of players. Each player $i$ has a valuation
function $v_i:2^{[m]}\rightarrow \RRp$ that is normalized
($v_i(\emptyset) = 0$) and monotone ($v_i(A) \leq v_i(B)$ whenever
$A \sse B$). A feasible solution is an \emph{allocation}
$(S_1,\ldots,S_n)$, where $S_i$ denotes the items assigned to player
$i$, and $\set{S_i}_i$ are mutually disjoint subsets of $[m]$. Player $i$'s value for outcome $(S_1,\ldots,S_n)$ is equal to $v_i(S_i)$.  The
goal is to choose the allocation maximizing \emph{social welfare}:
$\sum_i v_i(S_i)$.

\subsection{Matroid Rank Sum Valuations}\label{sec:MRS}

We now define matroid rank sum valuations. Relevant concepts from matroid theory are reviewed in Appendix~\ref{app:matroids}. 
\begin{definition}\label{def:mrs}
A set function $v:2^{[m]} \to \RRp$ is a \emph{matroid rank sum (MRS)}
function if there exists a family of matroid rank functions
$u_1,\ldots,u_\kappa: 2^{[m]} \to \NN$, and associated non-negative
weights $w_1,\ldots,w_\kappa \in \RR^+$, such that $v(S) =
\sum_{\ell=1}^\kappa w_\ell u_\ell(S)$ for all $S \sse [m]$. 
\end{definition}
%
We do not assume any particular representation of MRS valuations, and
require only oracle access to 
their (expected) values on certain distributions (see Section
\ref{sec:lotteryval}). 
MRS functions include most concrete examples of monotone submodular
functions that appear in the literature -- this includes 
coverage functions\footnote{A \emph{coverage function} $f$ on ground set
  $[m]$ designates some set $\L$ of elements, and $m$ subsets $A_1,\ldots,A_m \sse
  \L$, such that $f(S) = | \union_{j \in S} A_j|$. We note that
  $\L$ may be an infinite, yet measurable, space. Coverage functions
  are arguably {\em the} canonical example of a submodular
  function, particularly for combinatorial auctions.}, matroid weighted-rank
functions\footnote{This is a generalization of matroid rank functions,
  where weights are placed on elements of the matroid. It is true,
  though not immediately obvious, that a matroid weighted-rank function can be
  expressed as a weighted combination of matroid (unweighted) rank
  functions -- see e.g. \cite{revenuesubmod}.}, and all convex
combinations thereof.  Moreover, as shown in \cite{KhotLMM08}, $1-1/e$
is the best approximation possible in polynomial time for combinatorial auctions with MRS valuations unless $P = NP$, even ignoring strategic considerations. That being said, we note that some interesting submodular functions  --- such as some budget additive functions\footnote{  A set function $f$ on ground set $[m]$  is \emph{budgeted additive} if
  there exists a constant $B\geq 0$ (the budget) such that $f(S) =
  \min(B,\sum_{j \in S} f(\set{j}))$.
} --- are not in the matroid rank sum family (see Appendix~\ref{sec:beyondMRS}). 


\subsection{Lotteries and  Oracles}
\label{sec:lotteryval}

A \emph{value oracle} for a valuation $v: 2^{[m]} \to \RR$ takes as
input a set $S \sse [m]$, and returns $v(S)$.  We define an analogous
oracle that takes in a description of a simple lottery  over subsets of $[m]$, and outputs the expectation of $v$ over this lottery. 
Given a vector
$x\in [0,1]^m$ of probabilities on the items, let $D_x$ be the distribution over  $S \sse [m]$ that includes  each item $j$ in  $S$ independently with probability $x_j$. 
We use $F_v(x)$ to denote the expected value of $v(S)$ over draws $S \sim
D_x$ from this lottery.  

\begin{definition}
  A \emph{lottery-value oracle} for set function $v: 2^{[m]} \to \RR$
  takes as input a vector $x \in [0,1]^m$, and outputs
  \begin{equation}\label{eq:F} F_v(x) =\expect[S \sim D_x]{v(S)} =
  \sum_{S \sse [m]} v(S) \prod_{j \in S} x_j \prod_{j \neq S}
  (1-x_j).
\end{equation}  
\end{definition}
We note that $F_v$ is simply the well-studied \emph{multi-linear
  extension} of $v$ (see for example \cite{CCPV07,V08}).
In addition to being the
natural randomized analog of a value oracle, a lottery-value oracle is
easily implemented for various succinctly represented examples of MRS
valuations, like explicit coverage functions (see Appendix
\ref{app:separation}).

We also note that lottery-value oracle queries can be approximated arbitrarily well with high probability using a polynomial number of value oracle queries (see \cite{V08}).  Unfortunately, we are not able  to reconcile the incurred sampling errors --- small as they may be --- with the requirement that our mechanism be \emph{exactly} truthful. We suspect that relaxing our solution concept to approximate truthfulness -- also known as $\epsilon$-truthfulness -- would remove this difficulty, and allow us to relax our oracle model to the more traditional value oracles.

\section{Convex Rounding Framework}
\label{sec:framework}


\subsection{Relaxations and Rounding Schemes}

Let $\Pi$ be an optimization problem. A \emph{relaxation} $\Pi'$ of $\Pi$ defines for every $(\S,w) \in \Pi$
a convex and compact \emph{relaxed feasible set} $\R\sse \RR^m$ that
is independent of $w$ (we suppress the dependence on~$\S$);
and an \emph{extension} $w_\R: \R \to \RR$ of the objective $w$ to the
relaxed feasible set $\R$.  
This gives the following \emph{relaxed optimization problem}. 
\begin{lp}\label{lp:relaxation}
  \maxi{w_\R(x)}
\st
\con{x \in \R.}
\end{lp}
Generally, the extension is
defined so that it is computationally tractable to find a point $x \in
\R$ that maximizes $w_\R(x)$ (possibly approximately).

For example, $\S$ could be the allocations of~$m$ items to~$n$ bidders
in a combinatorial auction, $w(x)$  the welfare of an
allocation, $\R$ the feasible region of a linear
programming relaxation, and $w_{\R}$ the natural linear extension
of~$w$ to fractional allocations.


The solution $x \in \R$ to the relaxed problem need not be in
$\S$. 
A \emph{rounding scheme}  for relaxation $\Pi'$ of
$\Pi$ defines for each feasible set $\S$ of $\Pi$, and its
corresponding relaxed set  $\R$, a  (possibly randomized) function $r:
\R \to \S$. Since our rounding scheme will be randomized, we will frequently use $r(x)$ to denote the distribution over $\S$ resulting from rounding the point $x\in \R$.  Commonly, the rounding scheme satisfies the following
approximation guarantee: $\ex_{y \sim r(x)}[ w(y) ] \geq \alpha \cdot w_\R(x)$ for
every $x \in \R$. In this case, if $x^*$ maximizes $w_\R$ over $\R$
and $w_\R$ agrees with $w$ on $\S$, then
$\ex_{y \sim r(x^*)}[w(y)] \geq \alpha \cdot \max_{y \in \S}w(y)$.

\subsection{Convex Rounding Schemes and MIDR}

Our technique is motivated by the following observation: instead of
solving the relaxed problem and subsequently rounding the solution,
why not \emph{optimize directly on the outcome of the rounding
scheme}? 
In particular, consider the following relaxation 
 of $\Pi$
that ``absorbs'' rounding scheme $r$ into the objective. 
\begin{lp}
\label{lp:absorbrounding}
  \maxi{\ex_{y \sim r(x)}[ w(y) ]}
  \st
  \con{x \in \R.}
\end{lp}
The solution to this problem rounds to the best possible distribution in the range of the
rounding scheme, over all possible fractional solutions in $\R$. 
While this problem is often intractable, it always leads to an MIDR
allocation rule.

\begin{algorithm}
\caption{MIDR Allocation Rule via Optimizing over Output of Rounding Scheme}
\label{alg:midr}
\begin{algorithmic}[1]
\PARAMETER Feasible set $\S$ of $\Pi$.
\PARAMETER Relaxed feasible set $\R \sse \RR^m$.
\PARAMETER (Randomized) rounding scheme $r: \R \to \S$.
\INPUT Objective $w: \S \to \RR$ satisfying $(\S,w) \in \Pi$.
\OUTPUT Feasible solution $z \in \S$.
\STATE Let $x^*$ maximize ${\ex_{y \sim r(x)}[ w(y) ]}$ over
${x \in \R}$.
\STATE Let $z \sim r(x^*)$ \label{algstep:sample}
 \end{algorithmic}
\end{algorithm}

\begin{lemma} \label{lem:MIDR}
  Algorithm~\ref{alg:midr} is an MIDR allocation rule.
\end{lemma}

We say a rounding scheme $r: \R \to \S$ is
\emph{$\alpha$-approximate} for $\alpha \le 1$
if $w(x) \ge \ex_{y \sim r(x)}[w(y)] \geq \alpha \cdot w(x)$ for every $x \in \S$. 
When $r$ is $\alpha$-approximate, so is the allocation rule of
Algorithm~\ref{alg:midr}.
\begin{lemma} \label{lem:approx}
If $r$ is an $\alpha$-approximate rounding scheme, then Algorithm
\ref{alg:midr} returns an $\alpha$-approximate solution (in
expectation) to the original optimization problem~\eqref{lp:problem}.
\end{lemma}

For most rounding schemes
in the approximation algorithms literature, the
optimization problem~\eqref{lp:absorbrounding} cannot be
solved in polynomial time (assuming $P \neq NP$). 
The reason is that for any rounding scheme that always rounds a feasible
solution to itself -- i.e., $r(x)=x$ 
for all $x \in \S$ ---
an optimal solution to \eqref{lp:absorbrounding} is also optimal for
\eqref{lp:problem}. 
Thus, in this case, hardness of the original problem
\eqref{lp:problem} implies hardness of
\eqref{lp:absorbrounding}. 
We conclude that we need to design rounding schemes with the unusual
property that $r(x) \neq x$ for some $x\in \S$. 

We call a (randomized)  rounding scheme
$r: \R \to \S$  \emph{convex} if  $\ex_{y \sim r(x)}[ w(y) ]$ is concave
function of $x \in \R$.  
\begin{lemma}\label{lem:convexopt}
When $r$ is a convex rounding scheme for $\Pi'$,
\eqref{lp:absorbrounding} is a convex optimization problem.
\end{lemma}

Under additional technical conditions, discussed in the context of combinatorial auctions in Appendix~\ref{sec:solveCA}, the convex program~\eqref{lp:absorbrounding}
can be solved efficiently (e.g., using the ellipsoid method). 
This reduces the design of a polynomial-time $\alpha$-approximate MIDR
algorithm to designing a polynomial-time $\alpha$-approximate convex
rounding scheme.

Summarizing, Lemmas
\ref{lem:MIDR}, \ref{lem:approx}, and \ref{lem:convexopt} give the
following informal theorem. 
\begin{theorem}{(Informal)}
Let $\Pi$ be a welfare-maximization optimization problem, and let
$\Pi'$ be a relaxation of $\Pi$. If there exists a polynomial-time,
$\alpha$-approximate, convex rounding scheme for $\Pi'$, then there
exists a truthful-in-expectation, polynomial-time,
$\alpha$-approximate mechanism for $\Pi$.
\end{theorem}
Of course, there is no reason a priori to believe that useful
convex rounding schemes -- let alone ones computable in polynomial
time -- exist for any important problems. 
We show in Section~\ref{sec:CA} that they do in fact exist
and yield new results for an interesting class of combinatorial
auctions.

\section{Combinatorial Auctions}
\label{sec:CA}
\def\poiss{\qopname\relax n{poiss}}
\newcommand{\rp}{r_{\poiss}}

In this section, we use the framework of Section~\ref{sec:framework}
to prove our main result. 
\begin{thm}\label{thm:CAmain}
There is a $(1-1/e)$-approximate, truthful-in-expectation
mechanism for combinatorial auctions with matroid rank sum valuations in the
lottery-value oracle model, running in expected $\poly(n,m)$ time.
\end{thm}

We formulate welfare maximization in
combinatorial auctions as an optimization problem
$\Pi$. An instance $(\S,w) \in \Pi$ is given by the following integer program with feasible set $\S$ contained in
$\set{0,1}^{n \times m}$. Variable $x_{ij}$ indicates whether item $j$ is allocated to player $i$, and $w(x)$ denotes the social welfare of allocation $x$.
\begin{lp}\label{lp:CAinteger}
\maxi{w(x) = \sum_i v_i(\set{j: x_{ij}=1})}
\st
\qcon{\sum_{i} x_{ij} \leq 1}{j\in [m]}
\qcon{x_{ij} \in \set{0,1}}{i\in[n], j\in[m]}
\end{lp}
We let the relaxed feasible set $\R=\R(\S)$ be the result of relaxing the constraints $x_{ij} \in \set{0,1}$ of \eqref{lp:CAinteger} to $0 \leq x_{ij} \leq 1$. 

We structure the proof of Theorem~\ref{thm:CAmain} as follows. We
define the \emph{Poisson rounding scheme}, which we denote by $\rp$, in Section
\ref{subsec:convexCA}. We prove that $\rp$ is
$(1-1/e)$-approximate (Lemma~\ref{lem:poissonapprox}), and convex
(Lemma~\ref{lem:poissonconvex}).  Lemmas \ref{lem:MIDR},
\ref{lem:approx} and \ref{lem:poissonapprox}, taken together, imply
that Algorithm~\ref{alg:midr} when instantiated for combinatorial
auctions with $r=\rp$, is a $(1-1/e)$-approximate
MIDR allocation rule. Lemma~\ref{lem:poissonconvex} reduces implementing this allocation rule  to solving a convex program.

In Appendix~\ref{sec:solveCA}, we handle the  technical and
numerical issues related to solving convex programs. First, we prove
that our instantiation of Algorithm~\ref{alg:midr} for combinatorial
auctions can be implemented in expected polynomial-time using the
ellipsoid method under a simplifying assumption on the numerical
conditioning of our convex program (Lemma
\ref{lem:solveCAconditioned}). Then we show in Section
\ref{sec:noise} that the previous assumption can be
removed by slightly modifying our algorithm. 

Finally, we prove that truth-telling VCG payments can be computed
efficiently in Lemma~\ref{lem:compute_payments}.  
Taken together, these lemmas complete the proof of Theorem
\ref{thm:CAmain}. In Appendix~\ref{sec:beyondMRS}, we  discuss prospects for extending our result beyond matroid rank sum valuations.

\subsection{The Poisson Rounding Scheme}
\label{subsec:convexCA}


In this section we define the
\emph{Poisson rounding scheme}, which we denote by  $\rp$. The random map $\rp:\R\to\S$ renders the the following optimization problem over $\R$  a convex
optimization problem.  
\begin{lp}\label{lp:CArelaxation}
  \maxi{f(x)=\ex_{y \sim \rp(x)}[w(y)]}
\st
\qcon{\sum_{i} x_{ij} \leq 1}{j\in[m]}
\qcon{0 \leq x_{ij} \leq 1}{i\in[n], j \in [m]}
\end{lp}
We define the Poisson rounding scheme as follows. Given a fractional
solution $x$ to \eqref{lp:CArelaxation}, do the following
independently for each item $j$: assign $j$ to player $i$ with
probability $1-e^{-x_{ij}}$. 
(This is well defined since $1-e^{-x_{ij}} \le x_{ij}$ for all players $i$ and items $j$, and $\sum_i x_{ij} \leq 1$ for all items $j$.)
We make this more precise in Algorithm~\ref{alg:round}. For clarity, we represent an allocation as a function from items to players, with an additional null player $*$ reserved for items that are left unassigned.
\begin{algorithm}
\caption{The Poisson Rounding Scheme $\rp$}
\label{alg:round}
\begin{algorithmic}[1]
\INPUT Fractional allocation $x$ with $\sum_i x_{ij} \leq 1$ for all $j$, and $0 \leq x_{ij} \leq 1$ for all $i,j$.
\OUTPUT Feasible allocation $a:[m] \to [n] \union \set{*}$.
\FOR{$j=1,\ldots,m$}
\STATE Draw $p_j$ uniformly at random from $[0,1]$.
\IF{$\sum_i (1-e^{-x_{ij}}) \geq p_j$}
\STATE Let $a(j)$ be the minimum index such that $\sum_{i \leq a(j)} (1-e^{-x_{ij}}) \geq p_j$. 
\ELSE
\STATE $a(j) = *$
\ENDIF
\ENDFOR
 \end{algorithmic}
\end{algorithm}
The Poisson rounding scheme is $(1-1/e)$-approximate and convex. The
proof of Lemma~\ref{lem:poissonapprox}  is not difficult, 
and is included below. 
We prove Lemma~\ref{lem:poissonconvex} in
Section~\ref{sec:proveconvex}. As a warm-up, we first present a simplified proof of  Lemma~\ref{lem:poissonconvex} for the special case of coverage valuations in Section~\ref{sec:coverageconvex}.

\begin{lemma}\label{lem:poissonconvex}
 The Poisson rounding scheme is convex when player valuations are matroid rank sum  functions.  
\end{lemma}

\begin{lemma}\label{lem:poissonapprox}
The Poisson rounding scheme is $(1-1/e)$-approximate when player valuations
are submodular. 
\end{lemma}
\begin{proof}
Let $S_1,\ldots,S_n$ be an allocation, and let $x$ be an the integer point of $\eqref{lp:CArelaxation}$ corresponding to $S_1,\ldots,S_n$. Let $(S'_1,\ldots,S'_n) \sim \rp(x)$. It suffices to show that  $\ex[\sum_i v_i(S'_i)] \geq (1-1/e) \cdot \sum_i v_i(S_i)$. 

By definition of the Poisson rounding scheme, $S'_i$ includes each $j \in S_i$ independently with probability $1-1/e$. 
Submodularity implies that  $\ex[ v_i(S'_i)] \geq (1-1/e) \cdot  v_i(S_i)$ -- this was proved in many contexts: see for example \cite[Lemma 2.2]{FeigeMV07}, and the earlier related result in \cite[Proposition 2.3]{F06}. This completes the proof. 
\end{proof}

\subsection{Warm-up: Convexity for Coverage Valuations} 
\label{sec:coverageconvex}
In this section, we prove the special case of Lemma~\ref{lem:poissonconvex} for coverage valuations, as defined in Section~\ref{sec:MRS}. Fix $n$, $m$, and coverage valuations $\set{v_i}_{i=1}^n$, and let $\R$ denote the feasible set of mathematical program~\eqref{lp:CArelaxation}.  Let $(S_1,\ldots,S_n) \sim \rp(x)$ be the (random) allocation computed by the Poisson rounding scheme for point $x \in \R$. The expected welfare $\ex[w(\rp(x))]$ can be written as $\ex[\sum_{i=1}^n v_i(S_i)]$, where the expectation is taken over the internal random coins of the rounding scheme. By linearity of expectation, as well as the fact that the sum of concave functions is concave, it suffices to show that $\Ex[v_i(S_i)]$ is a concave function of $x$ for an arbitrary player $i$ with coverage valuation $v_i$. 

Fix player $i$, and use $x_{j}$, $v$, and $S$ as short-hand for $x_{ij}$, $v_i$, and $S_i$ respectively. Recall that $v$ is a coverage function; let $\L$ be a ground set and $A_1,\ldots,A_m \sse \L$ be such that $v_i(T) = |\union_{j \in T} A_j|$ for each $T \sse [m]$.  The Poisson rounding scheme includes each item $j$ in $S$ independently with probability $1-e^{-x_{j}}$. The expected value of player $i$ can be written as follows.
\begin{align*}
  \ex\left[ v(S) \right]  &= \ex[|\union_{j \in S} A_j|] \\
  &= \sum_{\ell \in \L} \pr [ \ell \in  \union_{j \in S} A_j] 
\end{align*}
Since the sum of concave functions is concave, it suffices to show that $\pr [ \ell \in  \union_{j \in S} A_j]$ is concave in $x$ for each $\ell \in \L$. We can interpret $\pr [ \ell \in  \union_{j \in S} A_j]$ as the probability that element $\ell$ is \emph{covered} by an item in $S$, where  $j \in [m]$ covers $\ell \in \L$ if $\ell \in A_j$. For each $\ell\in \L$, let $C_\ell$ be the set of items that cover $\ell$. Element $\ell \in \L$ is covered by $S$ precisely when $C_\ell \intersect S \neq \emptyset$. Each item $j \in C_\ell$ is included in $S$ independently with probability $1-e^{-x_{j}}$. Therefore, the probability $\ell \in \L$ is covered by $S$ can be re-written as follows: 
\begin{align}\label{prob:covered}
  \pr[\ell \in  \union_{j \in S} A_j] &= 1- \prod_{j \in C_\ell} e^{-x_{j}}\notag\\
 &=1-\exp\left(-\sum_{j \in C_\ell} x_{j}\right).
\end{align}
Form~\eqref{prob:covered} is the composition of the concave function $g(y)=1-e^{-y}$ with the affine function $y \to \sum_{j \in C_\ell} x_j$. It is well-known that composing a concave function with an affine function yields another concave function (see e.g. \cite{boyd}). Therefore, $\pr [ \ell \in  \union_{j \in S} A_j]$ is concave in $x$ for each $\ell \in \L$, as needed. This completes the proof.

\subsection{Convexity for Matroid Rank Sum Valuations}
\label{sec:proveconvex}
In this section, we will prove Lemma~\ref{lem:poissonconvex} in its full generality.  First, we define a discrete analogue of a Hessian matrix for set functions, and show that these discrete Hessians are negative semi-definite for matroid rank sum functions.
\begin{definition}\label{def:discrete-hess}
  Let $v:2^{[m]} \to \RR$ be a set function. For $S \sse [m]$, we define the \emph{discrete Hessian matrix} $\H^v_S \in \RR^{m \times m}$ of $v$ at $S$ as follows:
\begin{equation}
\label{eq:discretehess}
  \H^v_S(j,k) = v(S \union \set{j,k}) -v(S \union \set{j}) - v(S \union \set{k})  + v(S)
\end{equation}
for $j, k \in [m]$.
\end{definition}

\begin{claim}\label{claim:discrete-convex}
  If $v:2^{[m]} \to \RR^+$ is a matroid rank sum function, then $\H^v_S$ is negative semi-definite for each $S \sse [m]$.
\end{claim}
\begin{proof}
 We observe that $\H^v_S$ is linear in $v$, and recall that a non-negative weighted-sum of negative semi-definite matrices is negative semi-definite. Therefore, it is sufficient to prove this claim when $v$ is a matroid rank function.

Let $v$ be the matroid rank function of some matroid $M$ with ground set $[m]$, and fix $S\sse [m]$. Observe that $v$ is monotone, submodular, integer-valued, and $v(T \union \set{j}) \leq v(T) + 1$ for all $T \sse [m]$ and $j \in [m]$. Therefore, a simple case analysis reveals that for each $j,k \in [m]$

\begin{equation*}
\H^v_S(j,k)= 
\begin{cases} -1 &  \text{if }  v(S \union \set{j}) = v(S \union \set{k}) =    v(S \union \set{j,k}) = v(S) + 1,\\
0 &\text{otherwise.}
\end{cases}
\end{equation*}
In other words, $-\H^v_S$ is a binary matrix where $-\H^v_S(j,k)=1$ if and only if two conditions are satisfied: (1) Both $\set{j}$ and $\set{k}$ are independent sets in the contracted matroid $M / S$, and (2) $\set{j,k}$ is  dependent  in $M / S$.

It is clear that $-\H^v_S$ is symmetric. We now also show that   $-\H^v_S$ encodes a \emph{transitive} relation on $[m]$ --- i.e. for all $j, k, \ell \in [m]$, if $-\H^v_S(j,k)=-\H^v_S(k,\ell)=1$ then  $-\H^v_S(j,\ell)=1$. Fix $j, k, \ell$ such that $-\H^v_S(j,k)=-\H^v_S(k,\ell)=1$.  The sets $\set{j}$, $\set{k}$, and $\set{\ell}$ are independent sets of the contracted matroid $M / S$, and moreover $\set{j,k}$ and $\set{k,\ell}$ are dependent in $M / S$.  Assume for a contradiction that $\set{j,\ell}$ is independent in $M/S$;  applying the matroid exchange property to $\set{k}$ and $\set{j,\ell}$ implies that one of $\set{j,k}$ and $\set{k,\ell}$ must be independent in $M/S$ as well, contradicting our choice of $j$, $k$, and $\ell$. Therefore, $\set{j,\ell}$ is dependent in $M/S$, and $-\H^v_S(j,\ell)=1$.

 A binary matrix encoding a symmetric and transitive relation is a block diagonal matrix where each diagonal block is an all-ones or all-zeros sub-matrix. It is known, and easy to prove, that such a matrix is positive semi-definite. Therefore $\H^v_S$ is negative semi-definite.
\end{proof}

 We now return to Lemma~\ref{lem:poissonconvex}. Fix $n$, $m$, and MRS valuations $\set{v_i}_{i=1}^n$, and let $\R$ denote the feasible set of mathematical program~\eqref{lp:CArelaxation}.  Let $(S_1,\ldots,S_n) \sim \rp(x)$ be the (random) allocation computed by the Poisson rounding scheme for point $x \in \R$. The expected welfare $\ex[w(\rp(x))]$ can be written as $\ex[\sum_{i=1}^n v_i(S_i)]$, where the expectation is taken over the internal random coins of the rounding scheme. By linearity of expectation, as well as the fact that the sum of concave functions is concave, it suffices to show that $\Ex[v_i(S_i)]$ is a concave function of $x$ for an arbitrary player $i$ with MRS valuation $v_i$. 

Fix player $i$, and use $x_{j}$, $v$, $S$ as short-hand for $x_{ij}$, $v_i$, $S_i$ respectively. The Poisson rounding scheme includes each item $j$ in $S$ independently with probability $1-e^{-x_{j}}$. We can now write the expected value of player $i$ as the following function $G_v: \RR^m \to \RR$:
\begin{align}\label{eq:G}
G_v(x_1,\ldots,x_m) 
&= \sum_{S \sse [m]} v(S) \prod_{j \in S} (1-e^{-x_{j}}) \prod_{j \neq S} e^{-x_{j}}
\end{align}
The following claim, combined with Claim~\ref{claim:discrete-convex}, completes the proof of Lemma~\ref{lem:poissonconvex}.

\begin{claim}
  If all discrete Hessians of $v$ are negative semi-definite, then $G_v$ is concave.
\end{claim}
\begin{proof}
  Assume $\H^v_S$ is negative semi-definite for each $S \sse [m]$. We work with $G_v$ as expressed  in Equation~\eqref{eq:G}.  We will show that the Hessian matrix of $G_v$ at an arbitrary $x \in \RR^m$ is negative semi-definite, which is a sufficient condition for concavity. We take the mixed-derivative of $G_v$ with respect to $x_j$ and $x_k$ (possibly $j=k$).
\begin{align*}
  \frac{\partial^2 G_v(x)}{\partial x_j \partial x_k} =& \sum_{S \sse [m]\sm\set{j,k}}\ \prod_{\ell \in S} (1-e^{-x_\ell}) \prod_{\ell \in [m] \sm S} e^{-x_\ell} \bigg(v(S) - v(S \union \set{j}) - v(S \union \set{k}) + v(S \union \set{j,k})  \bigg)   \\
=& \sum_{S \sse [m]}\   \prod_{\ell \in S} (1-e^{-x_\ell}) \prod_{\ell \in [m] \sm S} e^{-x_\ell}   \bigg( v(S) - v(S \union \set{j}) - v(S \union \set{k}) + v(S \union \set{j,k})  \bigg)\\
=&  \sum_{S \sse [m]}\  \prod_{\ell \in S} (1-e^{-x_\ell}) \prod_{\ell \in [m]\sm S} e^{-x_\ell}\  \H^v_S(j,k)
\end{align*}
The first equality follows by grouping  the terms of Equation~\eqref{eq:G} by the projection of $S$ onto $[m] \sm \set{i,j}$, and then differentiating. The second equality follows from the fact that $v(S) - v(S \union \set{j}) - v(S \union \set{k}) + v(S \union \set{j,k}) = 0$ when $S$ includes either of $j$ and $k$.  The last equality follows by definition of $\H^v_S$.

The above derivation immediately implies that we can write the Hessian matrix of $G_v(x)$ as a non-negative weighted  sum of discrete Hessian matrices.
\begin{equation}
  \label{eq:4}
  \grad^2 G_v(x) =  \sum_{S \sse [m]} \  \prod_{\ell \in S} (1-e^{-x_\ell}) \prod_{\ell \in [m] \sm S} e^{-x_\ell} \ \H^v_S
\end{equation}
  A non-negative weighted-sum of negative semi-definite matrices is negative semi-definite. This completes the proof of the claim.
\end{proof}

\section*{Acknowledgments} We thank Ittai Abraham, Moshe Babaioff, Bobby Kleinberg, and Jan Vondr{\'a}k for helpful discussions and comments.  Specifically, we thank Bobby Kleinberg for insights into the algebraic properties of set functions that were useful in guiding the early stages of this work, and we thank Jan Vondr{\'a}k for pointing out that the proof of Lemma~\ref{lem:poissonconvex} can be simplified in the special case of coverage functions (Section~\ref{sec:coverageconvex}). Finally, we thank the anonymous STOC referees for helpful suggestions.

{
\bibliography{convex}
\bibliographystyle{amsplain}
}

\newpage
\appendix

\section{Combinatorial Auctions with Explicit Coverage Valuations} 
\label{app:separation}


In this section, we apply our mechanism to explicitly represented coverage valuations. This demonstrates the utility of our mechanism in a concrete, non-oracle-based setting, and moreover allows us to establish an interesting separation result. Specifically, we show that (1) The $(1-1/e)$-approximate  mechanism of Theorem~\ref{thm:CAmain} can be implemented in expected polynomial-time for this problem, and 
(2) No polynomial-time, universally-truthful, VCG-based\footnote{A universally-truthful mechanism is \emph{VCG-based} if it is a randomization over deterministic truthful mechanisms that each implement a \emph{maximal in range} allocation rule --- the special case of MIDR where each distribution in the distributional range is supported on a single allocation.} mechanism guarantees an approximation ratio of $o(n)$, unless $NP \sse P/poly$. The approximation ratio of $1-1/e$ is the best possible in polynomial-time for this problem --- even without incentive constraints --- assuming $P \neq NP$~\cite{KhotLMM08}. Ours is the first separation of its kind in the  \emph{computational complexity} model.\footnote{We note that this separation is meaningful because there are no known universally-truthful polynomial-time mechanisms --- VCG-based or otherwise --- for this problem that achieve an approximation ratio better than $\min(n,\sqrt{m})$. In particular, the result of \cite{D07} uses \emph{demand queries}, which can not be answered in polynomial time for explicit coverage valuations by the results of \cite{KhotLMM08} and \cite{FV06}.} 

 An $n$ player, $m$ item instance combinatorial auctions with \emph{explicit coverage valuations}  is described as follows. For each player $i$, there is a finite set $\L^i$, and a family $A^i_1,\ldots,A^i_m$ of subsets of $\L^i$. The valuation function of player $i$ is then defined as $v_i(S) = | \union_{j \in S} A^i_j |$. The set system $\left(\L^i, \set{A^i_j}_{j=1}^m\right)$ is encoded explicitly as a bipartite graph.

\subsection{A Truthful-in-Expectation Mechanism} 

As discussed previously, MRS valuations include all coverage valuations. Therefore, in order to implement the MIDR allocation rule of Section~\ref{sec:CA} for this problem, it suffices to  answer lottery-value queries in time polynomial in the number of bits encoding the instance.
\begin{claim}\label{claim:answer_succinct}
  In combinatorial auctions with explicit coverage valuations, lottery-value queries can be answered in time polynomial in the length of the encoding of the instance.
\end{claim}
\begin{proof}
Let $v:2^{[m]} \to \RRp$ be a coverage valuation presented explicitly as a set system $(\L, \set{A_j}_{j=1}^m)$, and let  $x\in [0,1]^m$.  Let $S$ be a random set that includes each $j \in [m]$ independently with probability $x_j$. The outcome of the lottery value oracle of $v$ evaluated at $x$ is equal to the sum, over all $\ell \in \L$, of the probability that $\ell$ is ``covered'' by  $S$ -- specifically, $\sum_{\ell \in \L} \pr[ \ell \in \union_{j \in S} A_j]$.  It is easy to verify that a term of this sum can be expressed as the following closed form expression. 
\[ \pr[\ell \in \union_{j \in S} A_j] =  1- \prod_{j : A_j \ni \ell} (1-x_j)\]
This expression can be evaluated in time polynomial in the representation of the set system. This completes the proof.
\end{proof}
\noindent Claim~\ref{claim:answer_succinct} implies the following Theorem.
\begin{theorem}
  There is an expected polynomial-time, $(1-1/e)$-approximate, truthful-in-expectation mechanism for combinatorial auctions with explicit coverage valuations.
\end{theorem}

\subsection{A Lower-bound on Universally Truthful VCG-Based Mechanisms}

We use the following special case of \cite[Theorem 1.2]{BDFKMPSSU10}: If a succinct combinatorial auction problem satisfies the \emph{regularity} conditions on the valuations defined in \cite{BDFKMPSSU10}, and moreover the 2-player version of the problem is APX hard, then no polynomial-time, universally-truthful, VCG-based mechanism guarantees an approximation ratio of $o(n)$.  

It is routine to verify the regularity assumptions of \cite{BDFKMPSSU10} for explicit coverage valuations.
APX-hardness of the 2-player problem follows by an elementary reduction from the APX-hard problem \emph{max-cut}. Given an instance of max-cut on a graph $G=(V,E)$, we let $[m]=V$, $\L^1=\L^2=E$. For $e\in E$, $i \in \set{1,2}$, and $j \in V$, we let $e \in A^i_j$ if $j$ is one of the endpoints of edge $e$. It is easy to check that the welfare maximizing allocation of the resulting $2$-player instance of combinatorial auctions corresponds to the maximum cut of $G$. Moreover, using the fact that the optimal objective value of max-cut is at least $|E|/2$, it is elementary to verify that the reduction preserves hardness of approximation up to a constant factor. Therefore, combinatorial auctions with explicit coverage valuations and 2 players is APX hard.  
This yields the following Theorem.
\begin{theorem}
  No universally truthful, polynomial-time, VCG-based mechanism for combinatorial auctions with explicit coverage valuations achieves a  approximation ratio of $o(n)$, unless $NP \sse P/poly$.
\end{theorem}

\section{Solving The Convex Program}
\label{sec:solveCA}

In this section, we overcome some technical difficulties related to the solvability of convex programs. We show in Section~\ref{sec:solveCAapprox} that, in the lottery-value oracle model, the four conditions for ``solvability'' of convex programs, as stated in Fact~\ref{fact:convex_solvability}, are easily satisfied for convex program~\eqref{lp:CArelaxation}. However, an additional challenge remains: ``solving'' a convex program -- as in Definition~\ref{def:rsolvable} -- returns an approximately optimal solution. Indeed the optimal solution of a convex program may be irrational in general, so this is unavoidable.


We show how to overcome this difficulty if we settle for polynomial runtime in expectation. While the optimal solution $x^*$ of \eqref{lp:CArelaxation} cannot be computed explicitly, the random variable $\rp(x^*)$ can be sampled in expected polynomial-time. The key idea is the following: \emph{sampling the random variable $\rp(x^*)$ rarely requires precise
 knowledge of $x^*$}.  Depending on the coin
 flips of $\rp$,  we decide how accurately we need to solve
 convex program~\eqref{lp:CArelaxation} in order compute $\rp(x^*)$.  Roughly speaking, we show that the probability of requiring a  $(1-\epsilon)$-approximation falls
exponentially in $\frac{1}{\epsilon}$. As a result, we can sample $\rp(x^*)$ in expected polynomial-time. We implement this plan in Section~\ref{sec:wellconditioned} under the simplifying assumption that convex program~\eqref{lp:CArelaxation} is \emph{well-conditioned} -- i.e. is ``sufficiently concave'' everywhere. In Section~\ref{sec:noise}, we show how to remove that assumption by slightly modifying our algorithm. 


\subsection{Approximating the Convex Program}
\label{sec:solveCAapprox}

\begin{claim}
\label{claim:solveCAapprox}
  There is an algorithm for Combinatorial Auctions with MRS valuations in the lottery-value oracle model that takes as input an instance of the problem and an approximation parameter $\epsilon >0$, runs in $\poly(n,m,\log(1/\epsilon))$ time, and returns a $(1-\epsilon)$-approximate solution to convex program $\eqref{lp:CArelaxation}$.
\end{claim}

It suffices to show that the four conditions of Fact~\ref{fact:convex_solvability} are satisfied in our setting. The first three are immediate from elementary combinatorial optimization (see for example \cite{schrijver}).
It remains to show that the first-order oracle, as defined in Fact~\ref{fact:convex_solvability}, can be implemented in polynomial-time in the lottery-value oracle model. The objective $f(x)$ of convex program~\eqref{lp:CArelaxation} can, by definition, be written as \[f(x)=\sum_i G_{v_i}(x_i),\]
where  $v_i$ is the valuation function of player $i$,  $x_i$ is the vector $(x_{i1},\ldots,x_{im})$, and and $G_{v_i}$ is as defined in \eqref{eq:G}. By definition,  $G_{v_i}(x_i)$ is the outcome of querying the lottery-value oracle of player $i$ with $(1-e^{-x_{i1}},\ldots,1-e^{-{x_{im}}})$ . Therefore, we can evaluate $f(x)$ using  $n$ lottery-value query, one for each player. It remains to show that we can also evaluate the (multi-variate) derivative $\grad f(x)$ of  $f(x)$. Using definition~\eqref{eq:G}, we take the partial derivative corresponding to $x_{ij}$. By rearranging the sum appropriately, we get that
\begin{align*}
\frac{\partial f}{\partial x_{ij}}(x) = e^{-x_{ij}}\bigg(F_{v_i}\big( (1-e^{-x_{i1}},\ldots,1-e^{-x_{im}}) \vee 1_j\big) - F_{v_i}\big((1-e^{-x_{i1}},\ldots,1-e^{-x_{im}})  \wedge 0_j\big)\bigg),
\end{align*}
where $F_{v_i}$ is as defined in Equation~\eqref{eq:F}. Here,  $\vee$ and $\wedge$ denote entry-wise minimum and maximum respectively, $1_j$ denotes the vector with all entries equal to $0$ except for a $1$ at position $j$, and $0_j$ denotes the vector with all entries equal to $1$ except for a $0$ at position $j$. 
It is clear that this entry of the gradient of $f$ can be evaluated using two lottery-value queries. Therefore,  $\grad f(x)$ can be evaluated using $2n$ lottery-value queries, $2$ for each player. This completes the proof of Claim~\ref{claim:solveCAapprox}.

\subsection{The Well-Conditioned Case}
\label{sec:wellconditioned}

In this section, we make the following simplifying assumption:  The objective function $f(x)$ of convex program~\eqref{lp:CArelaxation}, when restricted to any line in the feasible set $\R$, has a  second derivative of magnitude at least $\lambda=\frac{\sum_{i=1}^n v_i([m])}{ 2^{\poly(n,m)}}$ everywhere, where the polynomial in the denominator may be arbitrary. This is equivalent to requiring that every eigenvalue of the Hessian matrix of $f(x)$  has magnitude at least $\lambda$ when evaluated at any point in $\R$. Under this assumption, we prove Lemma~\ref{lem:solveCAconditioned}.
\begin{lemma}\label{lem:solveCAconditioned}
  Assume the magnitude of the second derivative of $f(x)$ is at least $ \lambda=\frac{\sum_{i=1}^n v_i([m])}{ 2^{\poly(n,m)}}$ everywhere. Algorithm~\ref{alg:midr}, instantiated for combinatorial auctions with $r=\rp$, can be simulated in time polynomial in $n$ and $m$ in expectation.
\end{lemma}


Let $x^*$ be the optimal solution to convex program~\eqref{lp:CArelaxation}. Algorithm~\ref{alg:midr} allocates items according to the distribution $\rp(x^*)$. The Poisson rounding scheme, as described in Algorithm~\ref{alg:round}, requires making $m$ independent decisions, one for each item $j$. Therefore, we fix item $j$ and show how to simulate this decision. It suffices to do the following in expected polynomial-time: flip uniform coin $p_j \in [0,1]$, and find the minimum index $a(j)$ (if any) such that $\sum_{i \leq a(j)} (1-e^{-x^*{ij}}) \geq p_j$. For most realizations of $p_j$, this can be decided using only coarse estimates $\tilde{x}_{ij}$ to $x^*_{ij}$. Assume we have an \emph{estimation oracle} for $x^*$ that, on input $\delta$, returns a \emph{$\delta$-estimate} $\tilde{x}$ of $x^*$: Specifically, $\tilde{x}_{ij} - x^*_{ij} \leq \delta$ for each $i$.  When $p_j$ falls outside the ``uncertainty zones'' of $\tilde{x}$, such as when $|p_j-\sum_{i' \leq i}(1-e^{-\tilde{x}_{i'j}})| > \delta n$ for each $i\in[n]$, it is easy to see that we can correctly determine $a(j)$ by using $\tilde{x}$ in lieu of $x$. The total measure of the uncertainty zones of $\tilde{x}$ is at most  $2n^2\delta$,  therefore $p_j$ lands outside the uncertainty zones with probability at least $1-2n^2\delta$. The following claim shows that if the estimation oracle for $x^*$ can be implemented in time polynomial in $\log(1/\delta)$, then we can simulate the Poisson rounding procedure in expected polynomial-time.
\begin{claim}
  Let $x^*$ be the optimal solution of convex program~\eqref{lp:CArelaxation}. Assume access to a subroutine $B(\delta)$ that returns a $\delta$-estimate of $x^*$ in time $\poly(n,m,\log(1/\delta))$. Algorithm~\eqref{alg:midr} with $r=\rp$ can be simulated in expected $\poly(n,m)$ time.
\end{claim}
\begin{proof}
  It suffices to show that we can simulate the allocation of an item $j$ by Algorithm~\eqref{alg:round} on input $x^*$. The simulation proceeds as follows: Draw $p_j \in [0,1]$ uniformly at random.  Start with $\delta=\delta_0=\frac{1}{2n^2}$. Let $\tilde{x}=B(\delta)$. While  $|p_j-\sum_{i' \leq i}(1-e^{-\tilde{x}_{i'j}})| \leq \delta n$ for some $i\in [n]$ (i.e. $p_j$ may fall inside an ``uncertainty zone'') do the following: let $\delta=\delta/2$, $\tilde{x}=B(\delta)$ and repeat. After the loop terminates, we have a sufficiently accurate estimate  of $x^*$ to calculate $a(j)$ as in Algorithm~\eqref{alg:round}.

It is easy to see that the above procedure is a faithful simulation of Algorithm~\eqref{alg:round} on $x^*$. It remains to bound its expected running time. Let $\delta_k=\frac{1}{2^{k+1}n^2}$ denote the value of $\delta$ at the $k$th iteration.  By assumption, the $k$th iteration takes $\poly(n,m,\log(1/\delta_k)) = \poly(n,m,\log(2^{k+1} n^2)) = \poly(n,m,k)$ time. The probability this procedure does not terminate after $k$ iterations is at most $2 n^2 \delta_k = 1/2^k$. Taken together, these two facts and a simple geometric summation imply that the expected runtime is polynomial in $n$ and $m$.
\end{proof}

It remains to show that the estimation oracle $B(\delta)$ can be implemented in $\poly(n,m,\log(1/\delta))$ time. At first blush, one may expect that the ellipsoid method can be used in the usual manner here. However, there is one complication: we require an estimate $\tilde{x}$ that is close to $x^*$ \emph{in solution space} rather than in terms of objective value. Using our assumption on the curvature of $f(x)$, we will reduce finding a $\delta$-estimate of $x^*$ to finding an $1-\epsilon(\delta)$ approximate solution to convex program~\eqref{lp:CArelaxation}. The dependence of $\epsilon$ on $\delta$ will be such that $\epsilon \geq \poly(\delta) / 2^{\poly(n,m)}$, thereby we can invoke Claim~\ref{claim:solveCAapprox} to deduce that $B(\delta)$ can be implemented in $\poly(n,m,\log(1/\delta))$ time.

Let $\epsilon=\epsilon(\delta) = \frac{\delta^2 \lambda}{2 \sum_i v_i([m]) }$. Plugging in the definition of $\lambda$, we deduce that $\epsilon \geq \delta^2 / 2^{\poly(n,m)}$, which is the desired dependence.  It remains to show that if $\tilde{x}$ is $(1-\epsilon)$-approximate solution to \eqref{lp:CArelaxation}, then $\tilde{x}$ is also a $\delta$-estimate of $x^*$.

 Using the fact that $f(x)$ is concave, and moreover its second derivative has magnitude at least $\lambda$, it a simple exercise to bound distance of any point $x$ from the optimal point $x^*$ in terms of its sub-optimality $f(x^*) - f(x)$, as follows: 
\begin{equation}
  \label{eq:distance_suboptimality}
f(x^*) - f(x)  \geq \frac{\lambda}{2} ||x - x^*||^2. 
\end{equation}
Assume $\tilde{x}$ is a $(1-\epsilon)$-approximate solution to \eqref{lp:CArelaxation}. Equation~\eqref{eq:distance_suboptimality} implies that
\begin{align*}
  ||\tilde{x} - x^*||^2 &\leq \frac{2}{\lambda} \epsilon f(x^*) 
  = \frac{\delta^2}{\sum_i v_i([m])} f(x^*) \leq \delta^2,
\end{align*}
where the last inequality follows from the fact that $\sum_i v_i([m]))$ is an upper-bound on the optimal value $f(x^*)$. Therefore, $||x - x^*|| \leq \delta$, as needed.
This completes the proof of Lemma~\ref{lem:solveCAconditioned}.


\subsection{Guaranteeing Good Conditioning}
\label{sec:noise}
\newcommand{\rpp}{\rp^+}
In this section, we propose a modification $\rpp$ of the Poisson rounding scheme $\rp$.  We will argue that $\rpp$ satisfies all the properties of $\rp$ established so far, with one exception: the approximation guarantee of Lemma~\ref{lem:poissonapprox} is reduced to $1-1/e -2^{-2mn}$. Then we will show that $\rpp$ satisfies the curvature assumption of Lemma~\ref{lem:solveCAconditioned}, demonstrating that said assumption may be removed.  Therefore Algorithm~\ref{alg:midr}, instantiated with $r=\rpp$ for combinatorial auctions with MRS valuations in the lottery-value oracle model, is $(1-1/e - 2^{-2mn})$ approximate and can be  implemented in expected $\poly(n,m)$ time. Finally, we show in Remark~\ref{rem:recoverapprox}  how to recover the $2^{-2mn}$ term to get a clean $1-1/e$ approximation ratio, as claimed in Theorem~\ref{thm:CAmain}.

Let 
$\mu= 2^{-2mn}$. 
 We define $\rpp$ in Algorithm~\ref{alg:round2}. Intuitively, $\rpp$ at first makes a tentative allocation using $\rp$. Then, it cancels said allocation with  small probability $\mu$. Finally, with probability $\beta$ it chooses a random ``lucky winner'' $i^*$ and gives him all the items. $\beta$ is defined as the fraction of items allocated in the original tentative allocation. The motivation behind this seemingly bizarre definition of $\rpp$ is purely technical: as we will see, it can be thought of as adding  ``concave noise'' to $\rp$.

\begin{algorithm}
\caption{Modified Poisson Rounding Scheme $\rpp$}
\label{alg:round2}
\begin{algorithmic}[1]
\INPUT Fractional allocation $x$ with $\sum_i x_{ij} \leq 1$ for all $j$, and $0 \leq x_{ij} \leq 1$ for all $i,j$.
\OUTPUT Feasible allocation $(S_1,\ldots,S_n)$.
\STATE Let $(S_1,\ldots,S_n) \sim \rp(x)$.
\STATE Let $\beta=\frac{\sum_i |S_i|}{m}$.
\STATE Draw $q_1 \in [0,1]$ uniformly at random.
\IF{$q_1 \in [0,\mu]$}
\STATE Let $(S_1,\ldots,S_n) = (\emptyset,\emptyset,\ldots,\emptyset)$.
\STATE Draw $q_2 \in [0,1]$ uniformly at random.
\IF{$q_2 \in [0,\beta]$}
\STATE Choose a player $i^*$ uniformly at random.
\STATE Let $S_{i^*}=[m]$, and $S_i=\emptyset$ for all $i \neq i^*$.
\ENDIF
\ENDIF
 \end{algorithmic}
\end{algorithm}

We can write the expected welfare $\ex[w(\rpp(x))]$ as follows. We use linearity of expectations and the fact that $\beta$ is independent of the choice of $i^*$ to simplify the expression.
\begin{align*}
\ex[w(\rpp(x))] & = \ex[ (1-\mu) w(\rp(x)) + \mu \beta v_{i*}([m]) ] \\
&= (1-\mu) \ex[w(\rp(x))] + \mu \ex[\beta] E[ v_{i*}([m])] \\
&= (1-\mu) \ex[w(\rp(x))] + \mu \ex[\beta] \frac{\sum_i v_i([m])}{n}
\end{align*}

Observe that $\rp$ allocates an item $j$ with probability $\sum_i (1-e^{-{x_{ij}}})$. Therefore, the expectation of $\beta$ is $\frac{\sum_{ij} (1-e^{-x_{ij}})}{m}$. This gives:
\begin{align}\label{eq:rpp_expectation}
 \ex[w(\rpp(x))] = &(1-\mu) \ex[w(\rp(x))] + \frac{\mu}{mn}  \sum_i v_i([m]) \sum_{i,j} (1-e^{-x_{ij}}).
\end{align}

It is clear that the expected welfare when using $r=\rpp$ is within $1-\mu = 1-2^{-2mn}$ of the expected welfare when using $r=\rp$ in the instantiation of Algorithm~\ref{alg:midr}. Using Lemma~\ref{lem:poissonapprox}, we conclude that $\rpp$ is a $(1-1/e - 2^{-2mn})$-approximate rounding scheme. Moreover, using Lemma~\ref{lem:poissonconvex}, as well as the fact that $(1-e^{-x_{ij}})$ is a concave function, we conclude that $\rpp$ is a convex rounding scheme. Therefore, this establishes the analogues of Lemmas~\ref{lem:poissonapprox} and\ref{lem:poissonconvex} for $\rpp$. It is elementary  to verify that our proof of Lemma~\ref{lem:solveCAconditioned} can be adapted to $\rpp$ as well.

It remains to show that $\rpp$ is ``sufficiently concave''. This would establish that the conditioning assumption of Section~\ref{sec:wellconditioned} is unnecessary for $\rpp$. We will show that expression~\eqref{eq:rpp_expectation} is a concave function with curvature of magnitude at least $\lambda=\frac{\sum_{i=1}^n v_i([m])}{ emn 2^{2mn}}$ everywhere. Since the curvature of concave functions is always non-positive, and moreover the curvature of the sum of two functions is the sum of their curvatures, it suffices to show that the second term of the sum~\eqref{eq:rpp_expectation} has curvature of magnitude at least $\lambda$. We note that the curvature of $\sum_{ij} (1-e^{-{x_{ij}}})$ is at least $e^{-1}$  over $x\in [0,1]^{n \times m}$. Therefore, the curvature of the second term of \eqref{eq:rpp_expectation} is at least  \[\frac{\mu}{mn} \left(\sum_i {v_i([m])}\right) e^{-1} = \lambda\] as needed.

\begin{remark}\label{rem:recoverapprox}
  In this section, we sacrificed $2^{-2mn}$ in the approximation ratio in order to guarantee expected polynomial runtime of our algorithm even when convex program~\eqref{lp:CArelaxation} is not well-conditioned. This loss can be recovered to get a clean $1-1/e$ approximation as follows. Given our $(1-1/e-2^{-2mn})$-approximate MIDR algorithm $\A$, construct the following algorithm $\A'$: Given an instance of combinatorial auctions, $\A'$ runs $\A$ on the instance with probability $1- e2^{-2mn}$, and with the remaining probability solves the instance optimally in exponential time $O(2^{2mn})$. It was shown in \cite{DR10} that a random composition of MIDR mechanisms is MIDR, therefore $\A'$ is MIDR. The expected runtime of $A'$ is bounded by the expected runtime of $\A$ plus $e 2^{-2mn} \cdot O(2^{2mn}) =O(1)$. Finally, the expected approximation  of $A'$ is the weighted average of the approximation ratio of $\A$ and the optimal approximation ratio $1$, and is at least $(1-e2^{-2mn})(1-1/e - 2^{-2mn}) + e2^{-2mn} \geq 1-1/e$.
\end{remark}


\section{Additional Preliminaries}

\subsection{Matroid Theory}
\label{app:matroids}
In this section, we review some basics of matroid theory. For a more comprehensive reference, we refer the reader to \cite{oxley}.   

A {\em matroid} $M$ is a pair $(\X,\I)$, where $\X$ is a finite \emph{ground set}, and $\I$ is a non-empty family of subsets of $\X$  satisfying the following two properties.  (1) {\em Downward closure:} If~$S$ belongs
to~$\I$, then so do all subsets of~$S$. (2)
{\em The Exchange Property:} Whenever $T,S \in \I$ with $|T| < |S|$,
there is some $x \in S \sm T$ such that $T \union \set{x} \in \I$. Elements of $\I$ are often referred to as  the \emph{independent sets} of the matroid. Subsets of $\X$ that are not in $\I$ are often called \emph{dependent}.

We associate with matroid $M$ a set function $rank_M:2^\X \to \NN$, known as the \emph{rank function of $M$}, defined as follows: $rank_M(A)= \max_{S \in \I} |S \intersect A|$. 
Equivalently, the rank of set $A$ in matroid $M$ is the maximum size of an independent set contained in $A$.  
A set function $f$ on a ground set $\X$ is a \emph{matroid rank function} if there exists a matroid $M$ on the same ground set such that $f=rank_M$. 
Matroid rank functions are monotone ($f(S) \leq f(T)$ when $S \sse T$), normalized ($f(\emptyset)=0$), and submodular ($f(S) + f(T) \geq f(S \intersect T) + f(S \union T)$ for all $S$ and $T$).

For a matroid $M=(\X,\I)$ and $S \sse \X$, we define the \emph{contraction} of $M$ by $S$, denoted by $M/S$. $M/S$ is a pair $(\X\sm S,\I')$, where $\I'$ is the following family of subsets of $\X \sm S$: A set $T \sse \X\sm S$ is in $\I'$   if and only if $rank_M(S\union T) - rank_M(S) = |T|$. For each matroid $M=(\X,\I)$ and $S \sse \X$, the contraction $M/S$ is also a matroid.

\subsection{Convex Optimization}
\label{sec:convexoptimization}

In this section, we distill some basics of convex optimization. For more details, see  \cite{nemirovski}.

\begin{definition}
  A maximization problem is given by a set $\Pi$ of instances $(\P,c)$, where $\P$ is a subset of some euclidean space, $c: \P \to \RR$, and the goal is to maximize $c(x)$ over $x \in \P$. We say $\Pi$  is a convex maximization problem if for every $(\P,c) \in \Pi$, $\P$ is a compact convex set, and $c: \P \to \RR$ is concave. If $c: \P \to \RR^+$ for every instance of $\Pi$, we say $\Pi$ is non-negative.
\end{definition}

\begin{definition}\label{def:rsolvable}
  We say a non-negative maximization problem $\Pi$ is \emph{$R$-solvable} in polynomial time if there is an algorithm that takes as input the representation of an instance $\I=(\P,c) \in \Pi$  --- where we use $|\I|$ to denote the number of bits in the representation --- and an approximation parameter $\epsilon$, and in time $\poly(|\I|,\log(1/\epsilon))$ outputs $x \in \P$ such that $c(x) \geq (1-\epsilon) \max_{y \in \P} c(y)$.
\end{definition}

\begin{fact}\label{fact:convex_solvability}
  Consider a non-negative convex maximization problem $\Pi$. If the following are satisfied, then $\Pi$ is $R$-solvable in polynomial time using the ellipsoid method. We let $\I=(\P,c)$ denote an instance of $\Pi$, and let $m$ denote the dimension of the ambient euclidean space.
  \begin{enumerate}
  \item Polynomial Dimension: $m$ is polynomial in $|\I|$.
  \item Starting ellipsoid: There is an algorithm that computes, in time $\poly(|\I|)$, a point $c \in \RR^m$, a matrix $A \in \RR^{m \by m}$, and a number $\V \in \RR$ such that the following hold. We use $E(c,A)$ to denote the ellipsoid given by center $c$ and linear transformation $A$.
    \begin{enumerate}
    \item  $E(c,A) \supseteq  \P$
    \item $\V \leq volume(\P)$
    \item $\frac{volume(E(c,A))}{\V} \leq 2^{\poly(|\I|)}$
    \end{enumerate}
  \item Separation oracle for $\P$: There is an algorithm that takes takes input $\I$ and $x \in \RR^m$, and in time $\poly(|\I|,|x|)$ where $|x|$ denotes the size of the representation of $x$, outputs ``yes'' if $ x\in \P$, otherwise outputs $h \in \RR^m$ such that $h^T x < h^T y$ for every $y \in \P$.
  \item First order oracle for $c$: There is an algorithm that takes input $\I$ and $x \in \RR^m$, and in time $\poly(|\I|,|x|)$ outputs $c(x) \in \RR$ and $\grad c ( x) \in \RR^m$. 

  \end{enumerate}

\end{fact}

\section{Additional Technical Details and Commentary}

\subsection{Computing Payments}
\label{sec:payments}

In this section, we show how to efficiently compute truth-telling payments for our mechanism. In fact, as shown below, this is possible for any maximal in distributional range allocation rule for combinatorial auctions given as a black box.

\begin{lemma}
  \label{lem:compute_payments}
Let $\A$ be an MIDR allocation rule for combinatorial auctions, and let $v_1, \ldots, v_n$ be input valuations. Assume  black-box access to $\A$, and value oracle access to $\set{v_i}_{i=1}^n$.  We can compute, with $\poly(n)$ over-head in runtime, payments $p_1,\ldots,p_n$ such that $\ex[p_i]$ equals the VCG payment of player $i$ for MIDR allocation rule $\A$ on input $v_1,\ldots,v_n$.
\end{lemma}
\begin{proof}
Without loss of generality, it suffices to show how to compute $p_1$. Let ${\bf 0} : 2^{[m]} \to \RR$ be the valuation evaluating to $0$ at each bundle. Recall (see e.g. \cite{Nis07}) that the VCG payment of player $1$  is equal to 
\begin{align}\label{eq:vcgpayment}
 \ex_{T \sim \A({\bf 0},v_2,\ldots,v_n)}\left[\sum_{i=2}^n v_i(T_i)\right] -\ex_{S \sim \A(v_1,\ldots,v_n)}\left[\sum_{i=2}^n v_i(S_i)\right]. 
\end{align}

Let $(S_1,\ldots,S_n)$~ be a sample from $\A(v_1,\ldots,v_n)$, and let ~$(T_1,\ldots,T_n)$~ be a sample from $\A({\bf 0},v_2,\ldots,v_n)$. Let $p_1= \sum_{i=2}^n v_i(T_i) - \sum_{i=2}^n v_i(S_i)$. Using linearity of expectations, it is easy to see that the expectation of $p_1$ is equal to the expression in \eqref{eq:vcgpayment}. This completes the proof.
\end{proof}

We note that the mechanism resulting from Lemma~\ref{lem:compute_payments} is individually rational in expectation, and each payment is non-negative in expectation. We leave open the question of whether it is possible to enforce individual rationality and non-negative payments for our mechanism ex-post.

\subsection{Beyond Matroid Rank Sum Valuations}
\label{sec:beyondMRS}
In this section, we discuss the prospect of extending our result beyond matroid rank sum valuations. First, we argue that our restriction to a subset of submodular functions is not merely an artifact of our analysis. Specifically, we exhibit a submodular  function that is not in the matroid rank sum family, and moreover the Poisson rounding scheme can be non-convex when a player has this function as their valuation. Then, we briefly argue that our mechanism may yet apply to some valuations that are not matroid rank sums.

We define a budget additive function $v$ on four items $\set{1,2,3,4}$. Three of the items are ``small'', one item is ``big'', and the budget equals the value of the big item. 
\begin{equation*}
v(S)=
\begin{cases} 
1 &  \text{if }  S=\set{j} \text{ for } j \in \set{1,2,3},\\
2 &  \text{if }  S=\set{4},\\
\min\left(\sum_{j \in S} v(\set{j}), 2\right) \ & \text{otherwise}
\end{cases}
\end{equation*}
We can show that $v$ is not a matroid rank sum function by invoking Claim~\ref{claim:discrete-convex}. Specifically, one can manually check that the discrete Hessian matrix $\H^v_\emptyset$ of $v$ at $\emptyset$ (see Definition~\ref{def:discrete-hess}) is not negative semi-definite. Moreover, for a player with valuation $v$, Poisson rounding renders the player's expected value function $G_v(x)$ (Equation~\eqref{eq:G})  non-concave in $x$: By Equation~\eqref{eq:4}, the Hessian matrix of $G_v(x)$ approaches the  discrete Hessian $\H^v_\emptyset$  as $x$ tends to zero. Since $\H^v_\emptyset$ is not negative semi-definite, $G_v(x)$ is non-concave for $x$ near zero. We note that we can construct a large family of similar counter examples by simply increasing the number of ``small items'' in $v$. 

Finally, we observe that our mechanism may apply to some valuations that are not matroid rank sums. We observe that we only used two properties of MRS functions: their discrete Hessian matrices are negative semi-definite (Claim~\ref{claim:discrete-convex}, which is used to prove Lemma~\ref{lem:poissonconvex}), and they are submodular (used to prove Lemma~\ref{lem:poissonapprox}). Therefore, our result extends directly to the class of all set functions satisfying both of these properties. We leave open the question of whether there exist interesting functions in this class that are not matroid rank sums. More generally, understanding the class of set functions with negative semi-definite discrete Hessian matrices --- in particular the relationship of this class to other classes of set functions studied in the literature --- may be an interesting direction for future inquiry.

\end{document}